\documentclass{article}%
\usepackage{amsfonts}
\usepackage{amsmath}
\usepackage{amssymb}
\usepackage{graphicx}%
\usepackage[hidelinks]{hyperref}%
\newtheorem{theorem}{Theorem}

\newenvironment{proof}[1][Proof]{\noindent\textbf{#1.} }{\ \rule{0.5em}{0.5em}}
\begin{document}

\title{Proofs of two Theorems concerning Sparse Spacetime Constraints}
\author{Christian Schulz\\MPI Informatik
\and Christoph von Tycowicz\\Freie Universit\"at Berlin
\and Hans-Peter Seidel\\MPI Informatik
\and Klaus Hildebrandt\\MPI Informatik }
\maketitle

\begin{abstract}
In the SIGGRAPH 2014 paper~\cite{Schulz2014} an approach for animating
deformable objects using sparse spacetime constraints is introduced. This
report contains the proofs of two theorems presented in the paper.

\end{abstract}


\section{Introduction}

In this report, we provide proofs of two theorems stated in \cite{Schulz2014}.
In Sections 2 and 4, we briefly review the background of the theorems and
introduce some notation. Sections 3 and 5 contain the proofs. For more
background on animating deformable objects using spacetime optimization, we
refer to
\cite{Witkin1988,Kass2008,Barbivc2009,Huang2011,Hildebrandt2012a,Barbic2012}.

\section{Sparse Constraints and Linear Dynamics}

We consider a linearized finite elements discretization of an elastic solid.
The dynamics of the solid are described by a coupled system of linear ordinary
second-order differential equations of the form%
\begin{equation}
M\,\ddot{u}(t)+(\alpha\,M+\beta\,K)\dot{u}%
(t)+K\,u(t)+g=0,\label{eq.linDynamics}%
\end{equation}
where $u\in%
\mathbb{R}
^{n}$ is the displacement vector, $M$ is the mass matrix, $K$ is the stiffness
matrix, $\alpha\,M+\beta\,K$ a Rayleigh damping term, and $g$ a constant
vector. We use spacetime constraints to force the object to interpolate a set
of keyframes. We will first look at the following simple set of keyframes. For
a set of $m+1$ nodes $\{t_{0},t_{1},\ldots,t_{m}\}$, we specify interpolation
constraints%
\begin{equation}
u(t_{i})=u_{i}\label{eq.keyframes}%
\end{equation}
and additionally the constraints
\begin{equation}
\dot{u}(t_{0})=v_{0}\qquad\text{and}\qquad\dot{u}(t_{m})=v_{m}%
\label{eq.boundaryConstraints}%
\end{equation}
on the velocity at the boundary of the time interval. To satisfy these
constraints, we need to inject an additional force to the system. This force
is determined in a optimization problem, where the objective functional is
\begin{equation}
E(u)=\frac{1}{2}%
{\displaystyle\int\limits_{t_{0}}^{t_{m}}}
\left\Vert M\,\ddot{u}+(\alpha\,M+\beta\,K)\dot{u}+K\,u+g\right\Vert _{M^{-1}%
}^{2}\,\text{d}t.\label{eq.spacetimeProblemLinear}%
\end{equation}
For some arbitrary $u$ the $E(u)$ measures the squared $L^{2}$-norm of the
additional force.

The eigenvalues and eigenmodes of~(\ref{eq.linDynamics}) are solutions to the
equation
\[
K\,\phi_{i}=\lambda_{i}M\,\phi_{i}.
\]
We consider a basis $\{\phi_{1},\phi_{2},...,\phi_{n}\}$ of $%
\mathbb{R}
^{n}$ consisting of eigenmodes. In \cite{Hildebrandt2012a} it was shown that
the minimizer $u$ of $E$ over all $\tilde{u}\in H^{2}([t_{0},t_{m}],%
\mathbb{R}
^{n})$ subject to the constraints (\ref{eq.keyframes}) and
(\ref{eq.boundaryConstraints}) are of the form
\begin{equation}
u(t)=\sum_{i}\omega_{i}(t)\phi_{i},\label{eq.wigglyMultivariate}%
\end{equation}
where the $\omega_{i}(t)$ are so-called \textit{wiggly splines}
\cite{Kass2008}. The wiggly splines are solutions to the one-dimensional form
of the spacetime optimization problem described above.

\subsection{Sparse spacetime constraints}

Instead of the interpolation constraints (\ref{eq.keyframes}) and
(\ref{eq.boundaryConstraints}), in \cite{Schulz2014} linear constraints of the
form
\[
A_{k}u(t_{k})=a_{k}\qquad\text{and}\qquad B_{k}\dot{u}(t_{k})-b_{k}%
\]
are considered. Here $A_{k},B_{k}$ are rectangular matrices and $a_{k},b_{k}$
are vectors. The constraints are sparse in the sense that the number of
constraints at each node $t_{k}$ is less than $n$. For example, only the
positions of a part of the object are prescribed.

Since the computation is, for efficiency, performed in a low-dimensional
subspace of $%
\mathbb{R}
^{n}$, the constraints are formulated as least squares constraints
\begin{equation}
E_{C}(u)=\frac{1}{2}\sum_{k=0}^{m}\left(  c_{A}\left\Vert A_{k}u(t_{k}%
)-a_{k}\right\Vert ^{2}+c_{B}\left\Vert B_{k}\dot{u}(t_{k})-b_{k}\right\Vert
^{2}\right)  , \label{eq.energyLS}%
\end{equation}
where $c_{A}\ $and $c_{B}$ are constants.

\section{The First Theorem}

The first theorem, \cite[Theorem 1]{Schulz2014}, shows that the minimizers of
the spacetime optimization problem with sparse (least squares) constraints can
be described using the eigenmodes and the wiggly splines.

\begin{theorem}
\label{thm.wiggliesLS}The minimizers of the energy $E(u)+E_{C}(u)$ among all
functions in the Sobolev space $H^{2}((t_{0},t_{m}),%
\mathbb{R}
^{n})$ are of the form~(\ref{eq.wigglyMultivariate}) and are twice
differentiable at any node $t_{k}$ where no velocity is prescribed and once
differentiable at all other nodes.
\end{theorem}

\begin{proof}
Assume that $u$ is a minimizer of $\mathcal{E}=E+E_{C}$ and that $v\in
H^{2}((t_{0},t_{m}),%
\mathbb{R}
^{n})$. The variation $\delta_{v}\mathcal{E}(u)$ of $\mathcal{E}$ at $u$ in
the direction of $v$ satisfies%
\[
\delta_{v}\mathcal{E}(u)=\delta_{v}E(u)+\delta_{v}E_{C}(u).
\]
The variation $\delta_{v}E_{C}(u)$ is
\begin{align*}
\delta_{v}E_{C}(u) &  =\underset{h\rightarrow0}{\lim}\frac{1}{h}\left(
E_{C}(u+hv)-E_{C}(u)\right)  \\
&  =\underset{h\rightarrow0}{\lim}\frac{1}{2h}\left(  \sum_{k=0}^{m}\left(
c_{A}\left\Vert A_{k}(u(t_{k})+hv(t_{k}))-a_{k}\right\Vert ^{2}+c_{B}%
\left\Vert B_{k}\left(  \dot{u}(t_{k})+h\dot{v}(t_{k})\right)  -b_{k}%
\right\Vert ^{2}\right)  -E_{C}(u)\right)  \\
&  =\sum_{k=0}^{m}\left(  \left(  A_{k}v(t_{k})\right)  ^{T}\left(
A_{k}u(t_{k})-a_{k}\right)  +\left(  B_{k}\dot{v}(t_{k})\right)  ^{T}\left(
B_{k}\dot{u}(t_{k})-b_{k}\right)  \right)  .
\end{align*}
Next, we consider the energy $E$ and abbreviate $D=\alpha\,M+\beta\,K$. The
variation $\delta_{v}E(u)$ is given by%
\begin{align*}
\delta_{v}E(u) &  =\underset{h\rightarrow0}{\lim}\frac{1}{h}\left(
E(u+hv)-E(u)\right)  \\
&  =\underset{h\rightarrow0}{\lim}\frac{1}{2h}\left(
{\displaystyle\int\limits_{t_{0}}^{t_{m}}}
\left\Vert M\,\left(  \ddot{u}+h\ddot{v}\right)  +D\left(  \dot{u}+h\dot
{v}\right)  +K\left(  u+hv\right)  +g\right\Vert _{M^{-1}}^{2}%
\,dt-E(u)\right)  \\
&  =%
{\displaystyle\int\limits_{t_{0}}^{t_{m}}}
\left(  \ddot{v}^{T}+\dot{v}^{T}DM^{-1}+v^{T}KM^{-1}\right)  \left(  M\ddot
{u}+D\dot{u}+Ku+g\right)  \,dt\\
&  =%
{\displaystyle\int\limits_{t_{0}}^{t_{m}}}
v^{T}\left(  M\ddddot{u}+\left(  2K-DM^{-1}D\right)  \ddot{u}+KM^{-1}\left(
Ku+g\right)  \right)  \,dt\\
&  -\sum_{k=1}^{m}\left(  \dot{v}^{T}+v^{T}DM^{-1}\right)  \left(  M\ddot
{u}+D\dot{u}+Ku+g\right)  {\Huge |}_{t_{k-1}}^{t_{k}}+\sum_{k=1}^{m}%
v^{T}\left(  M\dddot{u}+D\ddot{u}+K\dot{u}\right)  {\Huge |}_{t_{k-1}}^{t_{k}%
}.
\end{align*}
In the last step, we decomposed the integral over $[t_{0},t_{m}]$ into a sum
of integrals over the intervals $[t_{k},t_{k+1}]$ and used integration by
parts twice for each of the summands. We write $u,v,$ and $g$ in the
eigenbasis $\{\phi_{1},\phi_{2},...,\phi_{n}\}$
\[
u(t)=\sum_{i}\omega_{i}(t)\phi_{i},\quad\quad v(t)=\sum_{i}v_{i}(t)\phi
_{i},\quad\quad g=\sum_{i}g_{i}\phi_{i}%
\]
to obtain%
\begin{align}
\delta_{v}E(u) &  =%
{\displaystyle\int\limits_{t_{0}}^{t_{m}}}
\sum_{i}v_{i}\left(  \ddddot{\omega}_{i}+2\left(  \lambda_{i}-2\delta_{i}%
^{2}\right)  \ddot{\omega}_{i}+\lambda_{i}\left(  \lambda_{i}\omega_{i}%
+g_{i}\right)  \right)  \,dt\nonumber\\
&  +\sum_{k=1}^{m}\sum_{i}\left(  v_{i}\left(  \dddot{\omega}_{i}+2\delta
_{i}\ddot{\omega}_{i}+\lambda_{i}\dot{\omega}_{i}\right)  -\left(  \dot{v}%
_{i}+2\delta_{i}v_{i}\right)  \left(  \ddot{\omega}_{i}+2\delta_{i}\dot
{\omega}_{i}+\lambda_{i}\omega_{i}+g_{i}\right)  \right)  {\Huge |}_{t_{k-1}%
}^{t_{k}}.\label{eq.varE2}%
\end{align}
The variation $\delta_{v}\mathcal{E}(u)$ vanishes for any $v$ because $u$ is a
minimizer of $\mathcal{E}$. From the calculation of $\delta_{v}E_{C}(u)$ we
see that $\delta_{v}E_{C}(u)$ depends only on the values of $u,\dot{u},v,$ and
$\dot{v}$ at the nodes $t_{k}$ (and is independent of the values $u,\dot
{u},v,$ and $\dot{v}$ take at any $t$ in one of the open intervals
$(t_{k},t_{k+1})$). Then the integrals
\[%
{\displaystyle\int\limits_{t_{0}}^{t_{m}}}
v_{i}\left(  \ddddot{\omega}_{i}+2\left(  \lambda_{i}-2\delta_{i}^{2}\right)
\ddot{\omega}_{i}+\lambda_{i}\left(  \lambda_{i}\omega_{i}+g_{i}\right)
\right)  \,dt
\]
must vanish for all $v_{i}\in H^{2}((t_{0},t_{m}),%
\mathbb{R}
)$. This implies
\[
\ddddot{\omega}_{i}+2\left(  \lambda_{i}-2\delta_{i}^{2}\right)  \ddot{\omega
}_{i}+\lambda_{i}\left(  \lambda_{i}\omega_{i}+g_{i}\right)  =0.
\]
The last equation is exactly the characterization of the wiggly splines, see
\cite[Equation (4)]{Schulz2014}. This shows that $u$ is of the form
(\ref{eq.wigglyMultivariate}).

The function $u$ is once differentiable at the nodes $t_{k}$ because any
function in $H^{2}((t_{0},t_{m}),%
\mathbb{R}
)$ is (by the Sobolev's embedding theorem) once continuously differentiable.
Now what remains is to show that $u$ is twice differentiable at nodes where no
velocity is specified. For this, we reorder the terms of (\ref{eq.varE2}):%
\begin{align*}
\delta_{v}E(u) &  =\sum_{k=1}^{m-1}\sum_{i}v_{i}(t_{k})\left(  \dddot{\omega
}_{i}(\underrightarrow{t_{k}})-\dddot{\omega}_{i}(\underleftarrow{t_{k}%
})-2\delta_{i}\left(  \ddot{\omega}_{i}(\underrightarrow{t_{k}})-\ddot{\omega
}_{i}(\underleftarrow{t_{k}})\right)  \right)  -\dot{v}_{i}(t_{k})\left(
\ddot{\omega}_{i}(\underrightarrow{t_{k}})-\ddot{\omega}_{i}(\underleftarrow
{t_{k}})\right)  \\
&  +\sum_{i}\left(  v_{i}\left(  \dddot{\omega}_{i}+2\delta_{i}\ddot{\omega
}_{i}+\lambda_{i}\dot{\omega}_{i}\right)  -\left(  \dot{v}_{i}+v_{i}%
2\delta_{i}\right)  \left(  \ddot{\omega}_{i}+2\delta_{i}\dot{\omega}%
_{i}+\lambda_{i}\omega_{i}+g_{i}\right)  \right)  {\Huge |}_{t_{0}}^{t_{m}}.
\end{align*}
Here $\ddot{\omega}_{i}(\underrightarrow{t_{k}})$ denotes the second
derivative at $t_{k}$ of the restriction of $\ddot{\omega}_{i}$ to the
interval $[t_{k-1},t_{k}]$, and $\ddot{\omega}_{i}(\underleftarrow{t_{k}})$
denotes the second derivative at $t_{k}$ of the restriction of $\ddot{\omega
}_{i}$ to the interval $[t_{k},t_{k+1}]$. If no velocity is prescribed at the
node $t_{k}$, then $\dot{v}_{i}(t_{k})\left(  \ddot{\omega}_{i}%
(\underrightarrow{t_{k}})-\ddot{\omega}_{i}(\underleftarrow{t_{k}})\right)  $
has to vanish for all $\dot{v}_{i}$. This implies $\ddot{\omega}%
_{i}(\underrightarrow{t_{k}})=\ddot{\omega}_{i}(\underleftarrow{t_{k}})$ for
all $i$. Hence, $u$ is twice differentiable at $t_{k}$.
\end{proof}

\section{Sparse Constraints and Warping}

\textit{Rotation strain warping} was introduced in \cite{Huang2011}. The goal
there was to remove linearization artifacts from the deformation describe by
the displacement$~u$. The warp map $W$ is a nonlinear map on the space of all
possible displacements~$u$. To integrate the warping into the spacetime
optimization framework described above the least squares energy
(\ref{eq.energyLS}) is replaced by the nonlinear least squares energy%
\[
E_{WC}(u)=\frac{1}{2}\sum_{k=0}^{m}\left(  c_{A}\left\Vert A_{k}%
W(u(t_{k}))-a_{k}\right\Vert ^{2}+c_{B}\left\Vert B_{k}\text{D}W\,\dot
{u}(t_{k})-b_{k}\right\Vert ^{2}\right)  .
\]
Then, the objective functional
\begin{equation}
E(u)+E_{WC}(u)\label{eq.spacetimeOptiProbWarped}%
\end{equation}
is minimized over the space of displacements. The resulting motion is then
warped minimizer $W(u(t))$.

\section{The Second Theorem}

The second theorem in \cite{Schulz2014} shows that the minimizers of the
nonlinear optimization problem can still be described using the eigenmodes and
the wiggly splines.

\begin{theorem}
\label{thm.wiggliesLSWarped}The minimizers of the energy $E(u)+E_{WC}(u)$
among all functions in the Sobolev space $H^{2}((t_{0},t_{m}),%
\mathbb{R}
^{n})$ are of the form~(\ref{eq.wigglyMultivariate}) and are twice
differentiable at any node $t_{k}$ where no velocity is prescribed and once
differentiable at all other nodes.
\end{theorem}

\begin{proof}
[Proof (Sketch)]The proof is similar to that of Theorem \ref{thm.wiggliesLS}.
So we only sketch the proof here. We first calculate the variation $\delta
_{v}E_{WC}(u)$%
\begin{align}
\delta_{v}E_{WC}(u) &  =\underset{h\rightarrow0}{\lim}\frac{1}{h}\left(
E_{WC}(u+hv)-E_{WC}(u)\right)  \nonumber\\
&  =\underset{h\rightarrow0}{\lim}\frac{1}{2h}(\sum_{k=0}^{m}(c_{A}\left\Vert
A_{k}W(u(t_{k})+hv(t_{k}))-a_{k}\right\Vert ^{2}\nonumber\\
&  +c_{B}\left\Vert B_{k}DW\left(  \dot{u}(t_{k})+h\dot{v}(t_{k})\right)
-b_{k}\right\Vert ^{2})-E_{WC}(u))\label{eq.VariEWC}\\
&  =\sum_{k=0}^{m}(\left(  A_{k}DW(v(t_{k}))\right)  ^{T}\left(
A_{k}W(u(t_{k}))-a_{k}\right)  \nonumber\\
&  +\left(  B_{k}D^{2}W(\dot{v}(t_{k}))\right)  ^{T}\left(  B_{k}DW(\dot
{u}(t_{k}))-b_{k}\right)  ).\nonumber
\end{align}
The last step used the Taylor expansion
\[
W(u(t_{k})+hv(t_{k}))=W(u(t_{k}))+hDW(v(t_{k}))+\mathcal{R}(h)
\]
and
\[
DW\left(  \dot{u}(t_{k})+h\dot{v}(t_{k})\right)  =DW\left(  \dot{u}%
(t_{k})\right)  +hD^{2}W\left(  \dot{v}(t_{k})\right)  +\mathcal{R}(h),
\]
where $\mathcal{R}(h)$ is a remainder term for which $\underset{h\rightarrow
0}{\lim}\frac{1}{h}\mathcal{R}(h)=0$.

The rest is as in the proof of Theorem \ref{thm.wiggliesLS}. We calculate the
variation $\delta_{v}E(u)$ of $E$ and represent $u$ and $v$ in the modal
basis. This yields (\ref{eq.varE2}). From (\ref{eq.VariEWC}) we see that the
variation $\delta_{v}E_{WC}(u)$ depends only on the values of $u,\dot{u},v,$
and $\dot{v}$ at the nodes $t_{k}$ (and is independent of the values
$u,\dot{u},v,$ and $\dot{v}$ take at any $t$ in one of the open intervals
$(t_{k},t_{k+1})$). As described in the proof of Theorem \ref{thm.wiggliesLS},
the minimizers are of the form (\ref{eq.wigglyMultivariate}) and they are
twice differentiable at any node $t_{k},$ where no velocity is prescribed, and
once differentiable at all other nodes.
\end{proof}

\bibliographystyle{alpha}
\bibliography{partial}

\newcommand{\etalchar}[1]{$^{#1}$}
\begin{thebibliography}{HSvTP12}

\bibitem[BdSP09]{Barbivc2009}
Jernej Barbi\v{c}, Marco da~Silva, and Jovan Popovi\'{c}.
\newblock Deformable object animation using reduced optimal control.
\newblock {\em ACM Trans. Graph.}, 28:53:1--53:9, 2009.

\bibitem[BSG12]{Barbic2012}
Jernej Barbi\v{c}, Funshing Sin, and Eitan Grinspun.
\newblock Interactive editing of deformable simulations.
\newblock {\em ACM Trans. Graph.}, 31(4), 2012.

\bibitem[HSvTP12]{Hildebrandt2012a}
Klaus Hildebrandt, Christian Schulz, Christoph von Tycowicz, and Konrad
  Polthier.
\newblock Interactive spacetime control of deformable objects.
\newblock {\em ACM Trans. Graph.}, 31(4):71:1--71:8, 2012.

\bibitem[HTZ{\etalchar{+}}11]{Huang2011}
Jin Huang, Yiying Tong, Kun Zhou, Hujun Bao, and Mathieu Desbrun.
\newblock Interactive shape interpolation through controllable dynamic
  deformation.
\newblock {\em IEEE Transactions on Visualization and Computer Graphics},
  17(7):983--992, 2011.

\bibitem[KA08]{Kass2008}
Michael Kass and John Anderson.
\newblock Animating oscillatory motion with overlap: wiggly splines.
\newblock {\em ACM Trans. Graph.}, 27(3):28:1--28:8, 2008.

\bibitem[SvTSH14]{Schulz2014}
Christian Schulz, Christoph von Tycowicz, Hans-Peter Seidel, and Klaus
  Hildebrandt.
\newblock Animating deformable objects using sparse spacetime constraints.
\newblock {\em ACM Trans. Graph.}, 33(4), 2014.

\bibitem[WK88]{Witkin1988}
Andrew Witkin and Michael Kass.
\newblock Spacetime constraints.
\newblock {\em Proc. of ACM SIGGRAPH}, 22:159--168, 1988.

\end{thebibliography}

\end{document}